\documentclass[a4paper,11pt,twoside]{article}
\usepackage[margin=1in]{geometry}
\usepackage{setspace}
\usepackage[english]{babel}
\usepackage{afterpage}
\usepackage{mathrsfs,algorithm, algorithmic, yfonts}
\usepackage{harvard}
\usepackage{float}
\usepackage{comment}
\usepackage{multirow}
\usepackage{tabu}
\usepackage{subfigure}
\usepackage{placeins} %Neu zum platzieren der Graphiken
\usepackage[]{amsmath}
\usepackage{amsthm}
\usepackage{fix-cm}
\usepackage[]{amssymb}
\usepackage[]{latexsym}
\usepackage[latin1]{inputenc}
\usepackage[right]{eurosym}
\usepackage[T1]{fontenc}
\usepackage[]{graphicx}
\usepackage[]{epsfig}
\usepackage{fancyhdr}
\usepackage{bbm}
\usepackage{bm}
\usepackage{pstricks}
\usepackage{multirow}

\makeatletter

\renewcommand{\p@enumi}{theenumi-}
\renewcommand{\@fnsymbol}[1]{\@alph{#1}}
\usepackage{tabularx}
\usepackage{rotating}
\usepackage{float}

\newcommand{\bbr}{\mathbb{R}}  %black board bold \mathbb{R}

\newcommand{\ci}{\citeasnoun}

\newcommand{\be}{\begin{equation}}
\newcommand{\ee}{\end{equation}}
\newcommand{\bew}{\begin{equation*}}
\newcommand{\eew}{\end{equation*}}

\newcommand{\var}{{\rm V@R}}
\newcommand{\avar}{{\rm AV@R}}

\newcommand{\reavar}{{\rm RecAV@R}}

\newcommand{\ba}{\begin{array}{ll}}
\newcommand{\bal}{\begin{array}{ll}}
\newcommand{\ea}{\end{array}}

\newcommand{\E}{\mathbb{E}}

\newcommand{\probp}{\mathbb{P}}

\newcommand{\R}{\mathbb{R}}
\newcommand{\N}{\mathbb{N}}

\newcommand{\cE}{{\mathcal{E}}}
\newcommand{\cF}{{\mathcal{F}}}

\newcommand{\cS}{{\mathcal{S}}}

\newcommand{\cD}{\mathcal{D}}

\newcommand{\cI}{{\mathcal{I}}}
\newcommand{\cL}{\mathcal{L}}
\newcommand{\cM}{\mathcal{M}}
\newcommand{\cP}{\mathcal{P}}

\newcommand{\cX}{{\mathcal{X}}}

\newcommand{\VaR}{\mathop {\rm V@R}\nolimits}

\newtheorem{theorem}{Theorem}
\newtheorem{thm}{Theorem}
\newtheorem{definition}[thm]{Definition}
\newtheorem{proposition}[thm]{Proposition}

\newtheorem{remark}[thm]{Remark}

\begin{document}

\title{
Robust Portfolio Selection\\ Under Recovery Average Value at Risk
}

\author{
Cosimo Munari\footnote{Center for Finance and Insurance and Swiss Finance Institute, University of Zurich, Plattenstrasse 14, 8032 Zurich, Switzerland.
e-mail:  {\tt cosimo.munari@bf.uzh.ch}.} \\[1.0ex] \textit{University of Zurich}
\and
Justin Pl\"{u}ckebaum\footnote{Leibniz University Hannover, Welfengarten 1, 30167 Hannover, Germany.
e-mail:  {\tt justin.plueckebaum@insurance.uni-hannover.de}.} \\[1.0ex] \textit{Leibniz University Hannover}
\and
Stefan Weber\footnote{House of Insurance \& Institute of Actuarial and Financial Mathematics, Leibniz University Hannover, Welfengarten 1, 30167 Hannover, Germany.
e-mail:  {\tt stefan.weber@insurance.uni-hannover.de}.} \\[1.0ex] \textit{Leibniz University Hannover}
}
\date{\today}

\maketitle

\begin{abstract}
We study mean-risk optimal portfolio problems where risk is measured by Recovery Average Value at Risk, a prominent example in the class of recovery risk measures. We establish existence results in the situation where the joint distribution of portfolio assets is known as well as in the situation where it is uncertain and only assumed to belong to a set of mixtures of benchmark distributions (mixture uncertainty) or to a cloud around a benchmark distribution (box uncertainty). The comparison with the classical Average Value at Risk shows that portfolio selection under its recovery version enables financial institutions to exert better control on the recovery on liabilities while still allowing for tractable computations.
\end{abstract}
\vspace{0.2cm}
\textbf{Keywords:} Robust portfolio management; risk measures; recovery average at risk; efficient frontier; mean-risk optimal portfolios.
	
%\include{titlepage}
	
%	\newpage
	
	%\pagenumbering{Roman}
	%\tableofcontents
	
	%	\newpage
	%\listoffigures
	%	\newpage
	
	%\addtocounter{table}{-1}
	%\begin{longtable}{p{3 cm}p{17 cm}}{\Large \textbf{Abkürzungsverzeichnis}}\endhead
	%	et al. &   et alia (und andere) \\
	%\end{longtable}
	%\newpage
	
	\onehalfspacing
	\pagenumbering{arabic}
	\setcounter{page}{1}
	
	\setlength{\parindent}{0em}
	
\section{Introduction}

Portfolio selection is one of the central topics in mathematical finance and has been extensively studied in the literature. Since the pioneering publications by \ci{Markowitz52}, \ci{sharpe1963simplified}, \ci{budgets1965valuation}, much attention has been devoted to optimal portfolio problems in a mean-risk framework, where the objective is to study portfolios of financial assets that maximize expected returns subject to a given risk control. As in every optimization problem, the key questions from a theoretical perspective are those about existence, uniqueness, stability, and explicit identification of optimal portfolios. This, of course, highly depends on the chosen risk measure as well as on the assumptions on the (joint) distribution of the various assets. At the beginning, the literature has almost exclusively used the variance of the aggregated portfolio as the underlying measure of risk. In more recent years, especially after the publication of \ci{ADEH99}, there has been growing interest in revisiting mean-risk portfolio problems replacing the variance with risk measures that were deemed to capture risk in a more appropriate form, e.g., by focusing on the tail distribution of aggregated portfolios only. The bulk of the literature has focused on Value at Risk (V@R) and Average Value at Risk (AV@R) and on their comparison; see, e.g., \ci{RU00}, \ci{basak2001value}, \ci{campbell2001optimal}, \ci{frey2002var}, \ci{RU02}, \ci{yiu2004optimal}, \ci{yamai2005value}, \ci{leippold2006equilibrium}, \ci{ciliberti2007feasibility}, \ci{doi:10.1287/opre.1070.0433}, \ci{pirvu09}. \ci{GUNDEL20081126} study risk constraints in terms of utility-based shortfall risk. While most of the initial literature worked under the basic assumption that the joint distribution of portfolio assets is known, the subsequent literature has expanded the scope of research to include situations where there is uncertainty about the joint dependence across assets. The corresponding robust optimal portfolio problems under dependence uncertainty have been studied, e.g., in \ci{GUNDEL20071663}, \ci{quaranta2008robust} and \ci{ZhuFu09}.

\smallskip

The goal of this note is to investigate optimal portfolio problems in a mean-risk framework where risk is measured by Recovery Average Value at Risk (RecAV@R). This is a prominent example of a recovery risk measure, a concept that has been recently introduced in \ci{MWW23}. As argued there, recovery risk measures are designed to complement standard risk measures used in solvency regulation by offering portfolio managers the ability to exert a tighter control on the recovery of liabilities. In this sense, recovery risk measures have natural applications to mean-risk portfolio problems in an asset-liability management setting, where the risk constraint plays, for example, the role of an external regulatory constraint that can be interpreted as a solvency capital requirement. In the case of AV@R, one can only ensure solvency on average in the worst, say, $1\%$ (as in the Swiss Solvency Test) or $2.5\%$ (as in Basel III) of scenarios, but this {\em per se} does not provide any information about the ability to cover any pre-specified fraction of liabilities. However, it clearly matters to liability holders, and regulators on their behalf, if, say, $95\%$ or only $5\%$ of liabilities is recovered in the case of insolvency. A recovery risk measure like RecAV@R can be employed to this effect. By definition, RecAV@R ensures that assets are sufficient to cover on average any pre-specified fraction $\lambda$ of liabilities in the worst $100\gamma(\lambda)\%$ of scenarios. The function $\gamma$ can be chosen to tailor the relevant size of the tail distribution depending on the size of liabilities to be recovered. In particular, it is reasonable to assume that $\gamma$ is increasing and $\gamma(1)$ coincides with a regulatory threshold like $0.01$ (as in the Swiss Solvency Test) or $0.025$ (as in Basel III) to make sure {\em a priori} that RecAV@R is more stringent than the AV@R used in insurance or banking regulation.

\smallskip

This note is organized as follows. In Section 2 we briefly review the definition and the main properties of RecAV@R. In Section 3 we focus on optimal portfolio problems under RecAV@R both without and with dependence uncertainty. The main contribution is to show, by means of suitable minimax theorems, that optimal portfolios can be determined by solving appropriate linear programming problems that are both conceptually and computationally akin to the problems studied by \ci{RU00}, \ci{RU02}, and \ci{ZhuFu09} in the setting of mean-$\avar$ portfolio selection. In Section 4 we apply our results to study optimal portfolios in two concrete case studies. The first case study shows that there can be a marked difference between optimal portfolios under $\avar$ and $\reavar$. More specifically, in the presence of a risk-free and a risky asset, there are realistic situations where it is optimal under $\avar$ to fully invest in the risky asset whereas the optimal holding in the risky asset is capped if $\reavar$ is used to measure risk. In the second case study we focus on the more computational aspects and show that, in the presence of two risky assets whose returns follow standard distributions encountered in applications, the determination of robust efficient frontiers under $\reavar$ is feasible and computationally similar to the one under $\avar$.

%%%%%%%%%%%%%%%%%%%%%%%%%%%%%

\section{Recovery Average Value at Risk}

In this section we recall the definition and the basic properties of the risk measure Recovery Average Value at Risk ($\reavar$) introduced in \ci{MWW23}, to which we refer for the relevant proofs and for additional details. In the next sections we will take up the study of mean-risk portfolio problems where risk is quantified by $\reavar$.

\smallskip

Let $(\Omega,\cF,\probp)$ be a probability space and denote by $L^0$ the vector space of Borel measurable functions $X:\Omega\to\R$ (modulo $\probp$-almost sure equality). Throughout the paper we assume that positive values of $X$ represent a profit or a positive balance whereas negative values of $X$ represent a loss or a negative balance. The {\em Value at Risk} ($\var$) of $X\in L^0$ at level $\alpha\in[0,1]$ is defined by
\[
\VaR_\alpha(X):=
\inf\{x\in\R \,; \ \probp(X+x<0)\leq\alpha\}.
\]
The {\em Average Value at Risk} ($\avar$) of $X\in L^0$ at level $\alpha\in[0,1]$ is defined by
\[
\avar_\alpha(X):=
\begin{cases}
\frac{1}{\alpha}\int_0^\alpha \VaR_\beta(X)d\beta, & \mbox{if} \ \alpha\in(0,1],\\
\inf\{x\in\R \,; \ \probp(X+x\ge0)=1\} & \mbox{if} \ \alpha=0.
\end{cases}
\]
%We drop the reference to $\probp$ if the probability measure is clear from the context.

\begin{definition}
\label{def:reavar}
Let $\gamma: [0,1] \to [0,1]$ be an increasing function. The \emph{Recovery Average Value at Risk} ($\reavar$) of $(X,Y)\in L^0\times L^0$ with level function $\gamma$ is defined by
\[
\reavar_\gamma (X,Y) := \sup_{\lambda \in [0,1]}\avar_{\gamma(\lambda)}(X+(1-\lambda)Y).
\]
\end{definition}

\smallskip

Clearly, $\reavar$ is an extension of $\avar$. Indeed, by taking a constant function $\gamma$, say $\gamma\equiv\alpha$ for some $\alpha\in[0,1]$, one easily verifies that for every $X$ and for every positive $Y$
\[
\reavar_\gamma (X,Y) = \avar_\alpha(X).
\]
The definition of $\reavar$ is motivated by the following application. Consider a financial firm with stylized balance sheet at a generic time $t$ given by

\smallskip

\begin{center}
\begin{tabular}{|c|c|}
\hline
\bf Assets & \bf Liabilities \\
\hline\hline
\multirow{2}{*}{$A_t$}&$L_t$\\
\cline{2-2}
 &$E_t=A_t-L_t$\\
\hline
\end{tabular}
\end{center}

\smallskip

The quantity $E_t$ represents the net asset value of the firm and can be either positive or negative depending on whether the asset value $A_t$ is larger than the liability value $L_t$ or not. In the typical setting of a one-year horizon there are two reference dates, $t=0$ (today) and $t=1$ (end of the year). In a risk-sensitive solvency framework, the firm is adequately capitalized if its {\em available capital} $E_0$ is larger than a suitable {\em solvency capital requirement} that depends on the size of $E_1$ and therefore captures the inherent risk in the evolution of the balance sheet. In practice, solvency capital requirements are determined by applying a suitable risk measure $\rho$ like $\VaR$ or $\avar$ to the variation\footnote{In practice, instead of $E_0$ the expectation of $E_1$, typically discounted, is frequently used, see \ci{HKW20} for a discussion.} in the net asset value $\Delta E_1:=E_1-E_0$. The corresponding {\em solvency test} therefore takes the form
\[
\rho(\Delta E_1)\le E_0.
\]
The risk measure $\reavar$ can be used to define a solvency test of this type. Indeed, if the random variables $X$ and $Y$ in Definition \ref{def:reavar} are interpreted, respectively, as the net asset value $E_1$ and liabilities $L_1$ in the firm's balance sheet, then we can design the solvency test
\begin{equation}
\label{reavar_solvency test}
\reavar_\gamma(\Delta E_1,L_1)\leq E_0.
\end{equation}
The financial interpretation is clear once we observe that \eqref{reavar_solvency test} is equivalent to requiring that
\[
\avar_{\gamma(\lambda)}(A_1-\lambda L_1)\leq0, \ \ \ \forall\,\lambda\in[0,1].
\]
In words, the firm is adequately capitalized with respect to $\reavar$ if, for every fraction $\lambda\in[0,1]$, a firm with assets $A_1$ and liabilities $\lambda L_1$ is solvent on average in the worst $100\gamma(\lambda)\%$ scenarios (under $\probp$). In particular, the firm must be solvent on average in the worst $100\gamma(1)\%$ scenarios (under $\probp$), showing that \eqref{reavar_solvency test} is more stringent than a standard $\avar$ test at level $\gamma(1)$. It therefore comes as no surprise that, in the special case where the level function $\gamma$ is constant, say $\gamma\equiv\alpha$ for some $\alpha\in[0,1]$, the test \eqref{reavar_solvency test} boils down to a standard $\avar$ test
\begin{equation}
\label{avar_solvency test}
\avar_\alpha(\Delta E_1)=\reavar_\gamma(\Delta E_1,L_1)\leq E_0 \ \iff \ \avar_\alpha(E_1)\le0.
\end{equation}
In this case the firm is adequately capitalized if it is solvent on average in the worst $100\alpha\%$ scenarios (under $\probp$). The flexibility added by \eqref{reavar_solvency test} to the standard test \eqref{avar_solvency test} is that one can control recovery on liabilities, which is not permitted by standard solvency capital requirements based on $\avar$. This control is made possible by prescribing, in principle for each recovery level $\lambda\in[0,1]$, a different tail threshold $\gamma(\lambda)$. In this sense, it is natural to assume, as in Definition \ref{def:reavar}, that $\gamma$ is an increasing function: When we target a higher recovery on liabilities, we require solvency over a larger portion of the tail of $\Delta E_1$. It should be noted that \eqref{reavar_solvency test} also allows to control the probability of recovering the pre-specified fractions of liabilities. This is because, for any given $\alpha$ level, $\avar$ is larger than $\VaR$ and therefore
\[
\reavar_\gamma(\Delta E_1,L_1)\leq E_0 \ \implies \ \VaR_{\gamma(\lambda)}(A_1-\lambda L_1)\leq0, \ \ \ \forall\,\lambda\in[0,1].
\]
Translated in terms of recovery probabilities, we obtain as claimed
\[
\reavar_\gamma(\Delta E_1,L_1)\leq E_0 \ \implies \ \probp(A_1<\lambda L_1)\le\gamma(\lambda), \ \ \ \forall\,\lambda\in[0,1].
\]
The next proposition records an equivalent formulation of $\reavar$ when the level function $\gamma$ is piecewise constant. In this case, $\reavar$ is especially tractable and it is precisely this type of level functions that will be later used in our numerical studies.

\begin{proposition}
\label{prop: parametric gamma avar}
For $n\in\N_0$ let $0\leq\alpha_1<\cdots<\alpha_{n+1}\le1$ and $0<r_1<\cdots<r_n<r_{n+1}=1$. Define a function $\gamma:[0,1]\to[0,1]$ by
\[
\gamma(\lambda)=
\begin{cases}
\alpha_1 & \mbox{if} \ 0\leq\lambda<r_1,\\
\alpha_2 & \mbox{if} \ r_1\leq\lambda<r_2,\\
 \ \vdots\\
\alpha_n & \mbox{if} \ r_{n-1}\leq\lambda<r_n,\\
\alpha_{n+1} & \mbox{if} \ r_n\leq\lambda\leq r_{n+1}=1.\\
\end{cases}
\]
For all $X\in L^0$ and $Y\in L^0_+$
\[
\reavar_\gamma(X,Y) = \max_{i=1,\dots,n+1}\avar_{\alpha_i}(X+(1-r_i)Y).
\]
\end{proposition}

We conclude this section by stating some basic properties of $\reavar$, which follow at once from well-known properties of $\avar$.

\begin{proposition}
\label{prop: properties reavar}
Let $\gamma:[0,1]\to[0,1]$ be increasing. The following properties hold:
\begin{enumerate}
  \item {\em Cash-invariance in the first component}: For all $X,Y\in L^0$ and $m\in\R$
\[
\reavar_\gamma(X+m,Y)=\reavar_\gamma(X,Y)-m.
\]
  \item \emph{Monotonicity}: For all $X_1,X_2,Y_1,Y_2\in L^0$ such that $X_1\geq X_2$ and $Y_1\geq Y_2$
\[
\reavar_\gamma(X_1,Y_1) \leq \reavar_\gamma(X_2,Y_2).
\]
  \item \emph{Subadditivity:} For all $X_1,X_2,Y_1,Y_2\in L^0$
\[
\reavar_\gamma(X_1+X_2,Y_1+Y_2) \leq \reavar_\gamma(X_1,Y_1)+\reavar_\gamma(X_2,Y_2).
\]
  \item \emph{Positive homogeneity:} For all $X,Y\in L^0$ and $a\ge0$
\[
\reavar_\gamma(aX,aY) = a\reavar_\gamma(X,Y).
\]
\end{enumerate}
\end{proposition}

%%%%%%%%%%%%%%%%%%%%%%%%%%%%%

\section{Optimal portfolio selection under $\reavar$}
\label{sect: portfolio selection}

Risk measures are an important instrument to limit downside risk in portfolio optimization problems. This idea is related to the classical portfolio problem studied in \ci{Markowitz52}, where the objective was to select portfolios of reference financial assets with the goal of maximizing the expected return of the portfolio without exceeding a pre-specified level of standard deviation. In this sense, optimal portfolios represent the best tradeoff between risk and return. A similar problem can be reformulated in the language of asset-liability management for a financial institution, in which case the risk constraint is interpreted as a regulatory constraint. Standard deviation is, however, not a good risk measure in this type of applications because it fails to disentangle upside and downside risk. For this reason, the subsequent literature has investigated the mean-risk problem under different choices of tail risk measures, including $\VaR$ and $\avar$; see, e.g., \ci{RU00}, \ci{basak2001value}, \ci{campbell2001optimal}, \ci{RU02}, \ci{ZhuFu09}. Special attention has been devoted to $\avar$ because the resulting problem becomes convex and allows to exploit the rich methodology of convex optimization to characterize optimal portfolios.

\smallskip

In this section, we study the mean-risk problem for a financial institution that is subject to solvency capital requirements expressed in terms of the convex recovery risk measure $\reavar$. Our goal is to characterize the corresponding optimal portfolios. This task is, at first sight, more challenging than under $\avar$ because its recovery counterpart is defined as a supremum of standard $\avar$'s, making the mean-risk problem mathematically more involved. With the help of a suitable minimax theorem, which we establish for this purpose, we can nevertheless reduce the problem and show that standard techniques from linear programming can be exploited to identify optimal portfolios.

\smallskip

We consider a financial institution with total budget $b>0$ at time $t=0$. The company can invest in $k=1,\dots,K$ financial assets whose prices at dates $t=0,1$ are described by $S^k_t$ and whose relative returns are denoted by $R^k$ so that $S^k_1 = S^k_0(1 + R^k)$. We assume that $R^1,\dots,R^K\in L^1$. For every $k=1, \dots, K$ the company invests a fraction $x^k\ge0$ of its total budget into asset $k$ so that $\sum_{k=1}^K x^k =1$. For later convenience we define
\[
\Delta^K := \left\{\bm{x}\in\R^K_+ \,; \ \sum_{k=1}^K x^k =1\right\}.
\]
We also set $\bm{R} = (R^1, \dots , R^K)^\top$ and $\bm{x} = (x^1, \dots, x^K)^\top$. The total asset value at time $t=1$ is thus equal to
$$
A_1(\bm{x}) := b\sum_{k=1}^K x^k   (1 + R^k) \; = \; b\left( 1 + \sum_{k=1}^K x^k   R^k \right).
$$
In addition, we suppose that the company's liabilities at time $t=1$ amount to a random fraction $Z$ of the initial budget, i.e., the liabilities are equal to $L_1:=bZ$. We assume that $Z\in L^1$. The net asset value of the company equals
$$
E_1(\bm{x}):=A_1(\bm{x})-L_1=b\left( 1+\sum_{k=1}^K  \, x^k  \, R^k-Z\right).
$$
The expected net asset value is therefore given by
$$
\E(E_1(\bm{x})) = b\left( 1+\sum_{k=1}^K  \, x^k  \, \E(R^k)-\E(Z) \right).
$$
The mean-risk problem can equivalently be stated either as the maximization of expected returns under a risk constraint or as the minimization of risk for a target expected return. We focus on the latter formulation. For a given level function $\gamma:[0,1]\to[0,1]$ and for given $a\in\bbr$, we are thus interested in the following problem:
\begin{align}
\label{eq: basic problem 0}
\min_{\bm{x}\in\Delta^K} \; & \reavar_\gamma(E_1(\bm{x}),L_1)\\
& \mbox{s.t.} \quad \E(A_1(\bm{x}))\ge a.\nonumber
\end{align}
It is convenient to formulate the constraint for the expected return instead of the expected asset value. Using the properties of $\avar$ recorded in Proposition \ref{prop: properties reavar} and filtering out all constant terms, we can equivalently focus on the following problem for given $\mu\in\R$:
\begin{align*}
\min_{\bm{x}\in\Delta^K(\mu)}\reavar_\gamma \left(\sum_{k=1}^K x^kR^k-Z,Z\right),
\end{align*}
where the set of admissible portfolios is defined by
\[
\Delta^K(\mu) := \left\{\bm{x}\in\Delta^K \,; \ \sum_{k=1}^K x^k\E(R^k)\ge\mu\right\}.
\]
We focus on the special case of piecewise-constant level functions $\gamma$ introduced in Proposition \ref{prop: parametric gamma avar}. In this case, $\reavar$ is a maximum of finitely many $\avar$'s and the optimal portfolio problem can be equivalently written for given $\mu\in\R$ as:
\begin{align*}
\min_{\bm{x}\in\Delta^K(\mu)}
\max_{i=1,\dots,n+1}\avar_{\alpha_i}\left(\sum_{k=1}^K x^kR^k-r_iZ\right).
\end{align*}
As a last step, we exploit the representation of $\avar$ established in \ci{RU00} and \ci{RU02} to conveniently reformulate the problem above. To this effect, for $i=1, \dots ,n+1$ and $\bm{x}\in\Delta^K$ we can write
\[
\avar_{\alpha_i}\left(\sum_{k=1}^K x^kR^k-r_iZ\right) = \min_{v\in\bbr}\Psi^i (\bm{x}, v),
\]
where the auxiliary function $\Psi^i(\bm{x},\cdot):\R\to\R$ is defined by
$$
\Psi^i(\bm{x},v) := \frac 1 {\alpha_i}\E\left( \max\left\{v- \sum_{k=1}^K x^k   R^k + r_i Z,0\right\}\right) - v.
$$
As a consequence, our original optimal portfolio problem can be equivalently reformulated into the following minimax problem for given $\mu\in\R$:
\begin{align}
\label{eq: basic problem}
\min_{\bm{x}\in\Delta^K(\mu)}
\max_{i=1,\dots,n+1}\min_{v\in\bbr}\Psi^i (\bm{x}, v).
\end{align}
At first sight, this optimization problem seems difficult to cope with because of the entanglement between minimization and maximization. The following theorem shows that, by appropriately increasing the dimensionality of the internal minimization problem, we can interchange the order of minimum and maximum, thereby reducing the problem of finding optimal portfolios to a tractable linear programming problem.

\begin{theorem}
\label{thm:minimax}
For every $\bm{x}\in\Delta^K$ the following minimax equality holds:
$$
\max_{i=1,\dots,n+1}\min_{v\in\bbr}\Psi^i (\bm{x}, v) = \min_{\boldsymbol{v}\in\bbr^{n+1}}\max_{i=1,\dots,n+1}\Psi^i (\bm{x}, v^i).
$$
In particular, problem \eqref{eq: basic problem} can be equivalently written as
\begin{align*}
\min_{\bm{x}\in\Delta^K(\mu)}
\min_{\boldsymbol{v}\in\bbr^{n+1}}\max_{i=1,\dots,n+1}\Psi^i (\bm{x}, v^i).
\end{align*}
\end{theorem}
\begin{proof}
It is known from \ci{RU00} that, for each $i=1,\dots,n+1$, the convex function $\Psi^i(\bm{x},\cdot)$ attains its minimum on the (nonempty) compact interval $[q_i^-,q_i^+]$, where $q_i^-$ and $q_i^+$ are the lower, respectively upper, $\alpha_i$-quantiles of $\sum_{k=1}^{K}x^k R^k-r_iZ$. The desired minimax equality therefore follows at once from Theorem \ref{theo: minimax} in the appendix.
\end{proof}

\smallskip

In view of Theorem~\ref{thm:minimax}, the problem of determining the portfolios with miminal risk for a fixed expected target return $\mu\in\R$ can be equivalently expressed as
\begin{align*}
\min_{(\bm{x},\mathbf{v},T)\in \Delta^K(\mu)\times \R^{n+1}\times \R}\left\{T \,; \ \Psi^i(\bm{x},{v}^i) \le T, \ i=1,\dots,n+1\right\}.
\end{align*}
The evaluation of the functions $\Psi^i$'s involves the calculation of an expected value. In typical real-world applications, this is performed through Monte Carlo simulation. If $(\mathbf{R}_1,Z_1),\dots,(\mathbf{R}_S,Z_S)$ are $S$ independent simulations of the pair $(\mathbf{R},Z)$, we obtain the associated problem
\begin{align*}
\min_{(\bm{x},\mathbf{v},T)\in \Delta^K(\mu)\times \R^{n+1}\times \R}\left\{T \,; \ \frac{1}{\alpha_i S}\sum_{s=1}^S \max\left\{v^i-\sum_{k=1}^{K}x^kR_s^k+r_iZ_s,0\right\}-v^i \le T, \ i=1,\dots,n+1\right\}.
\end{align*}
The original problem can eventually be formulated as a tractable linear program of the form
\begin{align*}
	\min T\quad& \\
	s.t.\quad
	& \frac{1}{\alpha_i S}\sum_{s=1}^{S} u_j^i-v^i\le T, \ \ \ i=1,\dots,n+1,\\
	& u_s^i\ge v^i-\sum_{k=1}^{K}x^kR^{k}_s+r_iZ_s, \ \ \ i=1,\dots,n+1, \ s=1,\dots,S,\\
	& u_s^i\ge 0, \ \ \ i=1,\dots,n+1, \ s=1,\dots,S,\\
	& \frac{1}{m}\sum_{s=1}^{S}\sum_{k=1}^K x^k R^{k}_s\ge \mu,\\
	& \sum_{k=1}^K x^k=1,\\
	& x^k\ge 0,\quad k=1,\dots,K,\\
    & T,v^1,\dots,v^{n+1}\in\R.
\end{align*}
\begin{comment}
According to Theorem~\ref{thm:minimax}, to find optimal portfolios one can equivalently study the solutions of the following optimization problem for given $\mu\in\R$:
$$
\min_{(\bm{x},\boldsymbol{v})\in \cD\times \bbr^{n+1}}\left\{
c\in\R \,; \
\max_{i=1,\dots,n+1}\Psi^i (\bm{x}, v^i) \leq c\right\},
$$
where the set of admissible portfolios is given by
\[
\mathcal{D} := \left\{\bm{x}\in\Delta^K \,; \ \sum_{k=1}^K x^k\E(R^k)\ge\mu\right\}.
\]
The computation of the functions $\Psi^i$'s requires the estimation of an expectation. In typical applications in practice, this can be effectively achieved through Monte Carlo simulation. If we denote by $(\boldsymbol{R}^1, Z^1),\dots,(\boldsymbol{R}^M, Z^M)$ independent simulations of the pair $(\boldsymbol{R},Z)$, then we can reformulate the problem above as a linear programming problem of the form
\begin{eqnarray*}
\min c &&\\
\mbox{s.t.}
&& \frac 1 {M \cdot \alpha_i}\sum_{m=1}^M  u^{i,m} - v \leq c, \quad i=1,\dots, n+1,\\
&& u^{i,m} \geq v^i- \sum_{k=1}^K x^k   R^{k,m} - r_i Z^m, \quad i=1,\dots, n+1, \ m=1,\dots, M,\\
&& u^{i,m}\geq 0, \quad i=1,\dots, n+1, \ m=1,\dots, M,\\
\mbox{over}
&& (\bm{x},\boldsymbol{v},c)\in \mathcal{D} \times \bbr^{n+1} \times \bbr.
\end{eqnarray*}
\end{comment}

%%%%%%%%%%%%%%%%%%%%%%%%%%%%%%

\section{Robust optimal portfolio selection under $\reavar$}

In this section we study the optimal portfolio problem under uncertainty about the underlying probabilistic model. We will show that, in spite of the added complexity, the problem can still be reduced to a tractable linear programming problem.

\smallskip

Throughout the section we fix a measurable space $(\Omega,\cF)$ and denote by $\cL^0$ the vector space of Borel measurable functions $X:\Omega\to\R$. The set of all probability measures on $(\Omega,\cF)$ is denoted by $\cP$. Throughout we use a superscript to make explicit the dependence of our risk measures on the chosen probability measure in $\cP$.

\begin{definition}
Let $\gamma: [0,1] \to [0,1]$ be an increasing function and $\cM\subset\cP$. The \emph{Worst-Case Recovery Average Value at Risk} of $(X,Y)\in\cL^0\times\cL^0$ with level function $\gamma$ and uncertainty set $\cM$ is defined by
\[
\reavar^\cM_\gamma (X,Y) := \sup_{\probp\in\cM}\reavar^\probp_\gamma(X,Y).
\]
\end{definition}

\smallskip

For a given level function $\gamma:[0,1]\to[0,1]$ and a given uncertainty set $\cM\subset\cP$, and for given $a\in\bbr$, we are interested in the following robust version of problem \eqref{eq: basic problem 0}:
\begin{align}
\label{eq: robust problem}
\min_{\bm{x}\in\Delta^K} \; & \reavar^\cM_\gamma(E_1(\bm{x}),L_1)\\
& \mbox{s.t.} \quad \inf_{\probp\in\cM}\E_\probp(A_1(\bm{x}))\ge a.\nonumber
\end{align}
In the sequel we specify our analysis to two ways to define the uncertainty set $\cM$, which have been applied in \ci{ZhuFu09} to the study of robust mean-risk portfolio problems where the reference risk measure is $\avar$.

\subsection{Mixture uncertainty}

In a first step, we assume the existence of a finite number of benchmark probability measures, denoted by $\probp_1,\dots,\probp_m\in\cP$, and consider all possible convex combinations mixing them. This corresponds to the uncertainty set
%joint probability law of the asset returns $R^1,\dots,R^K$ and the liabilities $Z$ is described by an uncertain convex combination of a given number of benchmark probability measures. Denoting the benchmark probability measures by $\probp_1,\dots,\probp_m\in\cP$ for $m\in\N$, we obtain
\[
\cM_{mix} := \left\{\sum_{j=1}^{m}\lambda_j \probp_j \,; \ \bm{\lambda}\in\Delta^m\right\}.
\]
Consider a piecewise constant level function $\gamma$ as in Proposition~\ref{prop: parametric gamma avar}. For $i=1,\dots,n+1$ and $\bm{x}\in\Delta^K$ and for $j=1,\dots,m$ define the auxiliary function $\Psi^i_j(\bm{x},\cdot):\R\to\R$ by
\[
\Psi^i_j(\bm{x},v) := \frac{1}{\alpha_i}\E_{\probp_j}\left( \max\left\{v- \sum_{k=1}^K x^k R^k + r_i Z,0\right\}\right) - v.
\]
Repeating the reasoning in Section~\ref{sect: portfolio selection} we can recast the robust portfolio problem \eqref{eq: robust problem} with uncertainty set $\cM_{mix}$ in the following equivalent form for given $\mu\in\R$:
\begin{align}
\label{eq: basic problem robust}
\min_{\bm{x}\in\Delta^K_{mix}(\mu)} &
\max_{\bm{\lambda}\in\Delta^m}\max_{i=1,\dots,n+1}\min_{v\in\bbr}\sum_{j=1}^m\lambda_j\Psi^i_j(\bm{x},v),
\end{align}
where the set of admissible portfolios is defined by
\[
\Delta^K_{mix}(\mu) := \left\{\bm{x}\in\Delta^K \,; \ \min_{j=1,\dots,m}\sum_{k=1}^Kx^k\E_{\probp_j}(R^k)\ge\mu\right\}.
\]
The next theorem shows that the maxima and minima appearing in problem \eqref{eq: basic problem robust} can be reordered and coupled, thereby reducing the problem to a tractable linear programming problem.

\begin{theorem}
\label{thm:minimax robust}
For every $\bm{x}\in\Delta^K$ the following minimax equality holds:
$$
\max_{\bm{\lambda}\in\Delta^m}\max_{i=1,\dots,n+1}\min_{v\in\bbr}\sum_{j=1}^m\lambda_j\Psi^i_j(\bm{x},v) = \min_{\boldsymbol{v}\in\bbr^{n+1}}\max_{j=1,\dots,m}\max_{i=1,\dots,n+1}\Psi^i_j(\bm{x},v^i).
$$
In particular, problem \eqref{eq: basic problem robust} can be equivalently written as
\begin{align*}
\min_{\bm{x}\in\Delta^K_{mix}(\mu)} &
\min_{\boldsymbol{v}\in\bbr^{n+1}}\max_{j=1,\dots,m}\max_{i=1,\dots,n+1}\Psi^i_j(\bm{x},v^i).
\end{align*}
\end{theorem}
\begin{proof}
By \ci{RU00}, for all $i=1,\dots,n+1$ and $j=1,\dots,m$ the convex function $\Psi^i_j(\bm{x},\cdot)$ attains its minimum on the (nonempty) compact interval $[q_{i,j}^-,q_{i,j}^+]$, where $q_{i,j}^-$ and $q_{i,j}^+$ are the lower, respectively upper, $\alpha_i$-quantiles under $\probp_j$ of the random variable $\sum_{k=1}^{K}x^k R^k-r_iZ$. For every $i=1,\dots,n+1$ and for every choice of $\bm{\lambda}\in\Delta^m$ the function $\sum_{j=1}^m\lambda_j\Psi^i_j(\bm{x},\cdot)$ must therefore attain its minimum in the same compact interval, namely
\[
\cI_i := \left[\min_{j=1,\dots,m}q_{i,j}^-,\max_{j=1,\dots,m}q_{i,j}^+\right].
\]
For every $i=1,\dots,n+1$ the sets $\Delta^m$ and $\cI_i$ are compact and convex and the function $(\boldsymbol{\lambda},v) \mapsto \sum_{j=1}^m\lambda_j\Psi^i_j(\bm{x},v)$ is linear in $\boldsymbol{\lambda}$ and convex in $v$. As a consequence, the minimax theorem in \ci{fan1953minimax} delivers for every $i=1,\dots,n+1$
\begin{align*}
\label{eq: robust minimax 1}
\max_{\bm{\lambda}\in\Delta^m}\min_{v\in\bbr}\sum_{j=1}^m\lambda_j\Psi^i_j(\bm{x},v) &= \max_{\bm{\lambda}\in\Delta^m}\min_{v\in\cI_i}\sum_{j=1}^m\lambda_j\Psi^i_j(\bm{x},v)
= \min_{v\in\cI_i}\max_{\bm{\lambda}\in\Delta^m}\sum_{j=1}^m\lambda_j\Psi^i_j(\bm{x},v) \\
&= \min_{v\in\bbr}\max_{\bm{\lambda}\in\Delta^m}\sum_{j=1}^m\lambda_j\Psi^i_j(\bm{x},v) = \min_{v\in\bbr}\max_{j=1,\dots,m}\Psi^i_j(\bm{x},v).
\end{align*}
Since we can always interchange two consecutive maxima, we infer that
\[
\max_{\bm{\lambda}\in\Delta^m}\max_{i=1,\dots,n+1}\min_{v\in\bbr}\sum_{j=1}^m\lambda_j\Psi^i_j(\bm{x},v) = \max_{i=1,\dots,n+1}\min_{v\in\bbr}\max_{j=1,\dots,m}\Psi^i_j(\bm{x},v).
\]
As for every $i=1,\dots,n+1$ the convex function $\max_{j=1,\dots,m}\Psi^i_j(\bm{x},\cdot)$ attains its minimum on $\cI_i$, a direct application of Theorem \ref{theo: minimax} in the appendix yields the desired minimax equality.
\end{proof}

\smallskip

In the spirit of Section \ref{sect: portfolio selection}, one can use Theorem \ref{thm:minimax robust} to conveniently reformulate the portfolio problem under mixture uncertainty as
\begin{align*}
	\min_{(\bm{x},\mathbf{v},T)\in \Delta^K_{mix}(\mu)\times \R^{n+1}\times \R}\left\{T \,; \ \Psi_j^i(\bm{x},v^i) \le T, \ j=1,\dots,m, \ i=1,\dots,n+1\right\}.
\end{align*}
By approximating the expected value in the functions $\Psi_j^i$'s using Monte Carlo simulation as before, the problem can again be written as a tractable linear programming problem.

\begin{comment}
In view of Theorem~\ref{thm:minimax robust}, to find optimal portfolios under mixture uncertainty one can equivalently study the solutions of the following optimization problem for given $\mu\in\R$:
$$
\min_{(\bm{x},\boldsymbol{v})\in \cD\times \bbr^{n+1}}\left\{
c\in\R \,; \
\max_{j=1,\dots,m}\max_{i=1,\dots,n+1}\Psi^i_j (\bm{x}, v^i) \leq c\right\},
$$
where the set of admissible portfolios is given by
\[
\mathcal{D} := \left\{\bm{x}\in\Delta^K \,; \ \min_{j=1,\dots,m}\sum_{k=1}^K x^k\E_{\probp_j}(R^k)\ge\mu\right\}.
\]
The computation of the functions $\Psi^i_j$'s requires the estimation of an expectation. In typical applications in practice, this can be effectively achieved through Monte Carlo simulation. If we denote by $(\boldsymbol{R}^1, Z^1),\dots,(\boldsymbol{R}^M, Z^M)$ independent simulations of the vector $(\boldsymbol{R},Z)$, then we can reformulate the problem above as a linear programming problem of the form
\begin{eqnarray*}
\min c &&\\
\mbox{s.t.}
&& \frac 1 {M \cdot \alpha_i}\sum_{m=1}^M  u^{i,m} - v \leq c, \quad i=1,\dots, n+1,\\
&& u^{i,m} \geq v^i- \sum_{k=1}^K x^k   R^{k,m} - r_i Z^m, \quad i=1,\dots, n+1, \ m=1,\dots, M,\\
&& u^{i,m}\geq 0, \quad i=1,\dots, n+1, \ m=1,\dots, M,\\
\mbox{over}
&& (\bm{x},\boldsymbol{v},c)\in \mathcal{D} \times \bbr^{n+1} \times \bbr.
\end{eqnarray*}
\end{comment}

%%%%%%%%%%%%%%%%%%%%%%%%%%%%

\subsection{Box uncertainty}
\label{sect: box uncertainty}

In a second step, we fix a benchmark probability measure $\probp\in\cP$ under which the random vector $\bm{S}=(R^1,\dots,R^K,Z)$ is discrete and takes the values $\bm{s_1},\dots,\bm{s_m}\in\R^{K+1}$. To simplify the notation, we set for every $j=1,\dots,m$
\[
\pi_j := \probp(\bm{S}=\bm{s}_j).
\]
We consider all possible probability measures under which $\bm{S}$ remains discrete and that are obtained by a slight perturbation of the reference probability measure $\probp$. The set of perturbation parameters is defined for given $\underline{\bm{\varepsilon}},\overline{\bm{\varepsilon}}\in\R^m$ such that $\underline{\bm{\varepsilon}}\le\overline{\bm{\varepsilon}}$ and $\bm{\pi}+\underline{\bm{\varepsilon}}\ge\bm{0}$ by
\[
\cE := \left\{\bm{\varepsilon}\in\R^m \,; \ \underline{\bm{\varepsilon}}\le\bm{\varepsilon}\le\overline{\bm{\varepsilon}}, \ \sum_{j=1}^m\varepsilon_j=0\right\}
\]
For every $\bm{\varepsilon}\in\cE$ we consider a probability measure $\probp_{\bm{\varepsilon}}\in\cP$ such that for $j=1,\dots,m$
\[
\probp_{\bm{\varepsilon}}(\bm{S}=\bm{s}_j) = \pi_j+\varepsilon_j.
\]
The corresponding uncertainty set is given by
\[
\cM_{box} := \{\probp_{\bm{\varepsilon}}\in\cP \,; \ \bm{\varepsilon}\in\cE\}.
\]
Consider a piecewise constant level function $\gamma$ as in Proposition~\ref{prop: parametric gamma avar}. For $i=1,\dots,n+1$ and $\bm{x}\in\Delta^K$ and for $\bm{\varepsilon}\in\cE$ define the auxiliary function $\Psi^i_{\bm{\varepsilon}}(\bm{x},\cdot):\R\to\R$ by
\[
\Psi^i_{\bm{\varepsilon}}(\bm{x},v) := \frac{1}{\alpha_i}\E_{\probp_{\bm{\varepsilon}}}\left( \max\left\{v- \sum_{k=1}^K x^k R^k + r_i Z,0\right\}\right) - v.
\]
Repeating the reasoning in Section~\ref{sect: portfolio selection} we can recast the robust portfolio problem \eqref{eq: robust problem} with uncertainty set $\cM_{box}$ in the following equivalent form for given $\mu\in\R$:
\begin{comment}
\begin{align}
\label{eq: basic problem robust box}
\min_{\bm{x}\in\Delta^K} &
\max_{\bm{\varepsilon}\in\cE}\max_{i=1,\dots,n+1}\min_{v\in\bbr}\Psi^i_{\bm{\varepsilon}}(\bm{x},v) \\
& \mbox{s.t.} \quad \min_{\bm{\varepsilon}\in\cE}\sum_{i=1}^Kx^i\E_{\probp_{\bm{\varepsilon}}}(R^k)\ge\mu. \nonumber
\end{align}
\end{comment}
\begin{align}
\label{eq: basic problem robust box}
\min_{\bm{x}\in\Delta^K_{box}(\mu)} &
\max_{\bm{\varepsilon}\in\cE}\max_{i=1,\dots,n+1}\min_{v\in\bbr}\Psi^i_{\bm{\varepsilon}}(\bm{x},v),
\end{align}
where the set of admissible portfolios is defined by
\[
\Delta^K_{box}(\mu) := \left\{\bm{x}\in\Delta^K \,; \ \min_{\bm{\varepsilon}\in\cE}\sum_{k=1}^Kx^k\E_{\probp_{\bm{\varepsilon}}}(R^k)\ge\mu\right\}
\]
Once again, the maxima and minima appearing in problem \eqref{eq: basic problem robust box} can be reordered to yield a tractable linear programming problem. This is recorded in the next result.

\begin{theorem}
\label{thm:minimax robust box}
For every $\bm{x}\in\bbr^K$ the following minimax equality holds:
$$
\max_{\bm{\varepsilon}\in\cE}\max_{i=1,\dots,n+1}\min_{v\in\bbr}\Psi^i_{\bm{\varepsilon}}(\bm{x},v) = \min_{\boldsymbol{v}\in\bbr^{n+1}}\max_{\bm{\varepsilon}\in\cE}\max_{i=1,\dots,n+1}\Psi^i_{\bm{\varepsilon}}(\bm{x},v^i).
$$
In particular, problem \eqref{eq: basic problem robust box} can be equivalently written as
\[
\min_{\bm{x}\in\Delta^K_{box}(\mu)}
\min_{\boldsymbol{v}\in\bbr^{n+1}}
\max_{\bm{\varepsilon}\in\cE}
\max_{i=1,\dots,n+1}\Psi^i_{\bm{\varepsilon}}(\bm{x},v^i).
\]
\end{theorem}
\begin{proof}
We mimic the argument in the proof of Theorem \ref{thm:minimax robust}. By \ci{RU00}, for all $i=1,\dots,n+1$ and $\bm{\varepsilon}\in\cE$ the convex function $\Psi^i_{\bm{\varepsilon}}(\bm{x},\cdot)$ attains its minimum on the (nonempty) compact interval $[q_{i,\bm{\varepsilon}}^-,q_{i,\bm{\varepsilon}}^+]$, where $q_{i,\bm{\varepsilon}}^-$ and $q_{i,\bm{\varepsilon}}^+$ are the lower, respectively upper, $\alpha_i$-quantiles under $\probp_{\bm{\varepsilon}}$ of the discrete random variable $\sum_{k=1}^{K}x^k R^k-r_iZ$. For every $i=1,\dots,n+1$ we can thus define
\[
\cI_i := \left[\min_{\bm{\varepsilon}\in\cE}q_{i,\bm{\varepsilon}}^-,\max_{\bm{\varepsilon}\in\cE}q_{i,\bm{\varepsilon}}^+\right].
\]
Clearly, for every $i=1,\dots,n+1$ the sets $\cE$ and $\cI_i$ are compact and convex and the function $(\bm{\varepsilon},v) \mapsto \Psi^i_{\bm{\varepsilon}}(\bm{x},v)$ is linear in $\bm{\varepsilon}$ and convex in $v$. As a consequence, the minimax theorem in \ci{fan1953minimax} delivers for every $i=1,\dots,n+1$
\begin{align*}
\label{eq: robust minimax 1 box}
\max_{\bm{\varepsilon}\in\cE}\min_{v\in\bbr}\Psi^i_{\bm{\varepsilon}}(\bm{x},v) &= \max_{\bm{\varepsilon}\in\cE}\min_{v\in\cI_i}\Psi^i_{\bm{\varepsilon}}(\bm{x},v)
= \min_{v\in\cI_i}\max_{\bm{\varepsilon}\in\cE}\Psi^i_{\bm{\varepsilon}}(\bm{x},v) \\
&= \min_{v\in\bbr}\max_{\bm{\varepsilon}\in\cE}\Psi^i_{\bm{\varepsilon}}(\bm{x},v).
\end{align*}
Since we can always interchange two consecutive maxima, we infer that
\[
\max_{\bm{\varepsilon}\in\cE}\max_{i=1,\dots,n+1}\min_{v\in\bbr}\Psi^i_{\bm{\varepsilon}}(\bm{x},v) = \max_{i=1,\dots,n+1}\min_{v\in\bbr}\max_{\bm{\varepsilon}\in\cE}\Psi^i_{\bm{\varepsilon}}(\bm{x},v).
\]
As for every $i=1,\dots,n+1$ the convex function $\max_{\bm{\varepsilon}\in\cE}\Psi^i_{\bm{\varepsilon}}(\bm{x},\cdot)$ attains its minimum on $\cI_i$, a direct application of Theorem \ref{theo: minimax} in the appendix yields the desired minimax equality.
\end{proof}

\smallskip

Similarly to the case of mixture uncertainty, Theorem \ref{thm:minimax robust box} can be used to conveniently reformulate the original portfolio problem as
\begin{align*}
	\min_{(\bm{x},\mathbf{v},T)\in \Delta^K_{box}(\mu)\times \R^{n+1}\times \R}\left\{T \,; \ \Psi_{\bm{\varepsilon}}^i(\bm{x},v^i) \le T, \ \bm{\varepsilon}\in\cE, \ i=1,\dots,n+1\right\}.
\end{align*}
Once again, by approximating the expected value in the functions $\Psi_{\bm{\varepsilon}}^i$'s using Monte Carlo simulation, the problem can be written as a tractable linear programming problem. The procedure is described in detail for the case where $n=1$ in the case study in Section~\ref{sect: efficient frontier}.

%%%%%%%%%%%%%%%%%%%%%%%%%%%%%%%%%%%%%

\section{Numerical illustrations}

This final section is devoted to an illustration of mean-risk portfolio selection under $\reavar$ in the context of two cases studies. In the first case study we compare optimal portfolios under $\avar$ and $\reavar$ and document that already the choice of a simple level function $\gamma$ may lead to a drastic difference in the composition of optimal portfolios. More specifically, in the presence of a risk-free and a risky asset, there are realistic situations where it is optimal to fully invest in the risky asset under $\avar$ whereas the optimal holding in the risky asset is capped if $\reavar$ is used to measure risk. In the second case study we focus on the more computational aspects and show that, in the presence of two risky assets whose returns follow standard distributions encountered in applications, the determination of robust efficient frontiers under $\reavar$ is feasible and computationally similar to the one under $\avar$. To achieve this, we exploit the minimax theorems established in the previous section and combine them with standard Monte Carlo simulation.

%%%%%%%%%%%%%%%%%%%%%%%%

\subsection{Case study 1: Optimal portfolio without dependence uncertainty}

We consider a financial institution with total budget $b > 0$ at the initial date. The management can invest in two assets with one-period relative returns given by
\[
R^1=0, \ \ \ \ \ \
R^2=\begin{cases}
0.5\% & \mbox{with probability} \ 99.9\%,\\
-4\% & \mbox{with probability} \ 0.1\%.
\end{cases}
\]
At the terminal date the institution is exposed to deterministic liabilities amounting to $L_1=b\ell$ for given $\ell\in(0,1)$. We denote by $x$ the fraction of total budget that is invested in the risky asset. The corresponding end-of-period net asset value is therefore equal to
\[
E_1(x) := b(1-x)(1+R^1)+bx(1+R^2)-L_1 = b(1+xR^2-\ell).
\]
Set $\alpha=1\%$ and for $\mu=0$ define the set of admissible holdings in the risky asset by
\[
\cX := \{x\in[0,1] \,; \ (1-x)\E(R^1)+x\E(R^2)\ge\mu\}.
\]
Since $\E(R^2)=0.004955$, we readily see that $\cX=[0,1]$. In a first step we focus on the problem
\begin{align}
\label{eq: illustration 1 avar}
\min_{x\in[0,1]}\avar_\alpha(E_1(x)).
\end{align}
Using the properties of $\avar$ we can equivalently write
\[
\min_{x\in[0,1]}\avar_\alpha(E_1(x))=b\left(\min_{x\in[0,1]}\{\avar_\alpha(R^2)x\}-1+\ell\right).
\]
This shows that the composition of the optimal portfolio will be driven by the sign of $\avar_\alpha(R^2)$. A direct computation shows that
\[
\avar_\alpha(R^2)=\frac{1}{\alpha}\left(\frac{1}{1000}\frac{4}{100}-\frac{9}{1000}\frac{5}{1000}\right)=-0.0005<0.
\]
As a result, problem \eqref{eq: illustration 1 avar} admits a unique optimal solution given by $x=1$. In words, the optimal portfolio under $\avar$ corresponds to investing the entire available budget into the risky asset. We turn to investigating how the optimal portfolio changes if $\avar$ is replaced by $\reavar$. To this effect, take $\beta\in(0,\alpha)$ and $r\in(0,1)$ and consider a simple level function of the form
\[
\gamma(\lambda):=
\begin{cases}
\alpha & \mbox{if} \ \lambda\in[r,1],\\
\beta & \mbox{if} \ \lambda\in[0,r).
\end{cases}
\]
We modify problem \eqref{eq: illustration 1 avar} by replacing $\avar$ with $\reavar$, thereby obtaining the new problem
\begin{align}
\label{eq: illustration 1 reavar}
\min_{x\in[0,1]}\reavar_\gamma(E_1(x),L_1).
\end{align}
In view of Proposition \ref{prop: parametric gamma avar}, we can equivalently write
\begin{align*}
\min_{x\in[0,1]}\reavar_\gamma(E_1(x),L_1) = \min_{x\in[0,1]}\max\{\avar_\alpha(E_1(x)),\avar_\alpha(E_1(x)+(1-r)L_1)\}.
\end{align*}
Using the properties of $\avar$ we obtain the more explicit problem
\begin{align*}
\min_{x\in[0,1]}\reavar_\gamma(E_1(x),L_1) = b\min_{x\in[0,1]}\max\{\avar_\alpha(R^2)x-1+\ell,\avar_\beta(R^2)x-1+r\ell\}.
\end{align*}
The optimal portfolio is thus determined by the sign of $\avar_\alpha(R^2)$ and $\avar_\beta(R^2)$. By design, we always have $\avar_\beta(R^2)>\avar_\alpha(R^2)$. Moreover, recall that $\avar_\alpha(R^2)<0$. A direct computation shows that
\[
\avar_\beta(R^2)=
\begin{cases}
\frac{4}{100} & \mbox{if} \ \beta\in(0,0.1\%],\\
\frac{45}{1000000}\frac{1}{\beta}-\frac{5}{1000} & \mbox{if} \ \beta\in(0.1\%,\alpha).
\end{cases}
\]
In particular, we have
\[
\avar_\beta(R^2)\le0 \ \iff \ \beta\ge\frac{45}{1000000}\frac{1000}{5}=0.9\%.
\]
\begin{figure}[t]
%\vspace{-0.8cm}
\centering
\subfigure{
\includegraphics[width=0.4\textwidth]{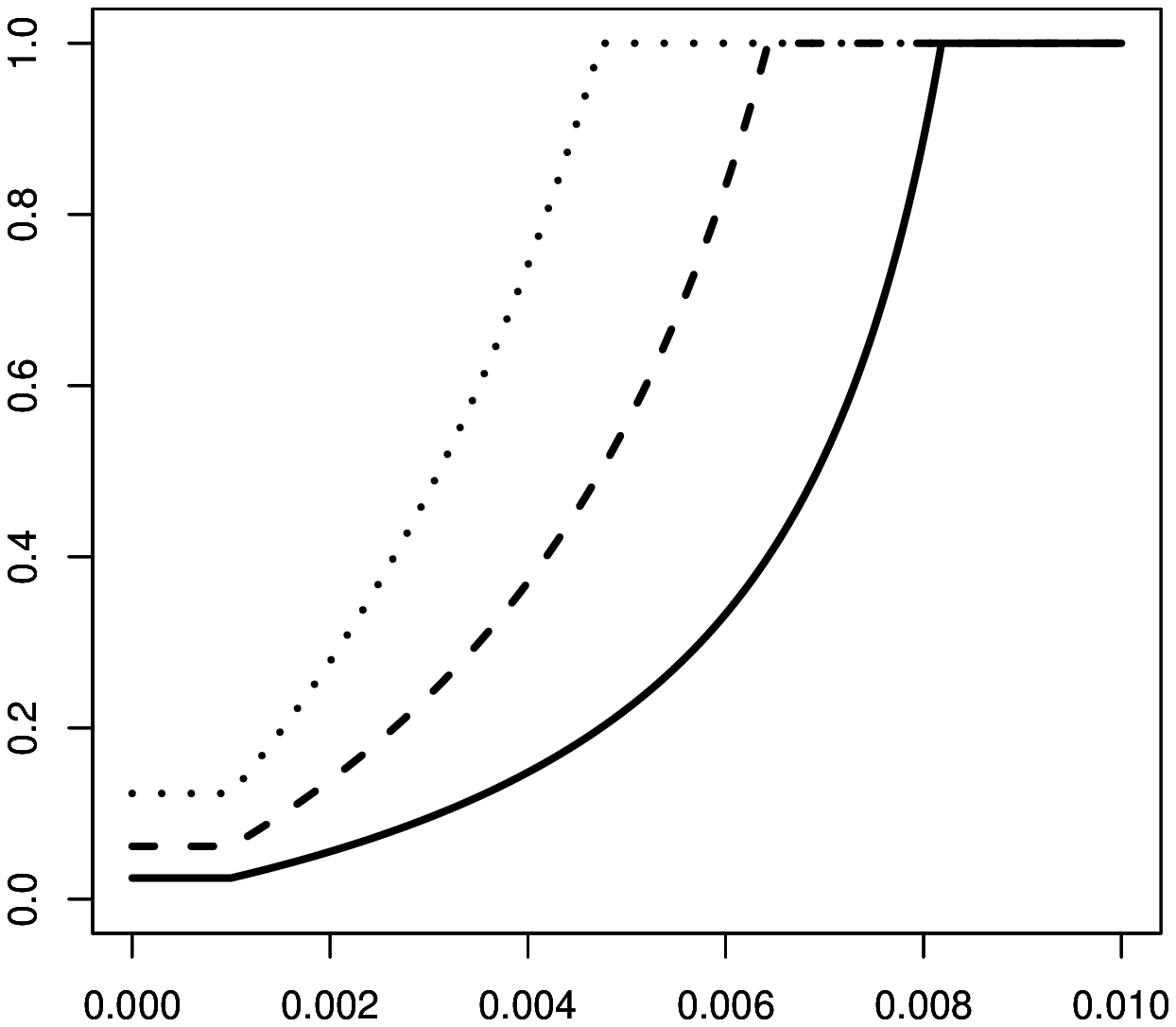}
}
\hspace{1cm}
\subfigure{
\includegraphics[width=0.4\textwidth]{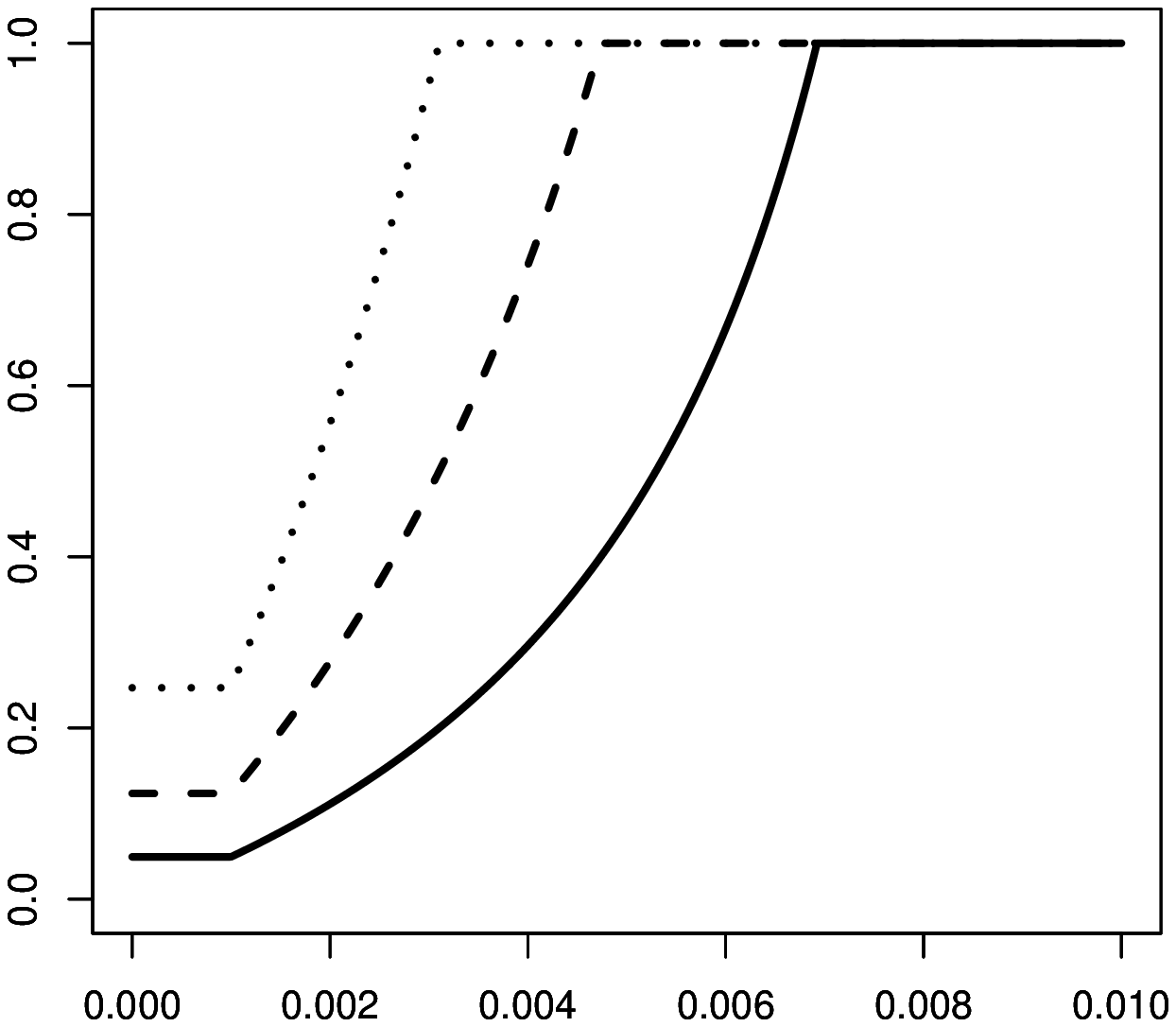}
}
%\vskip%
%\vspace{-1cm}%\baselineskip
\subfigure{
\hspace{0.15cm}\includegraphics[width=0.4\textwidth]{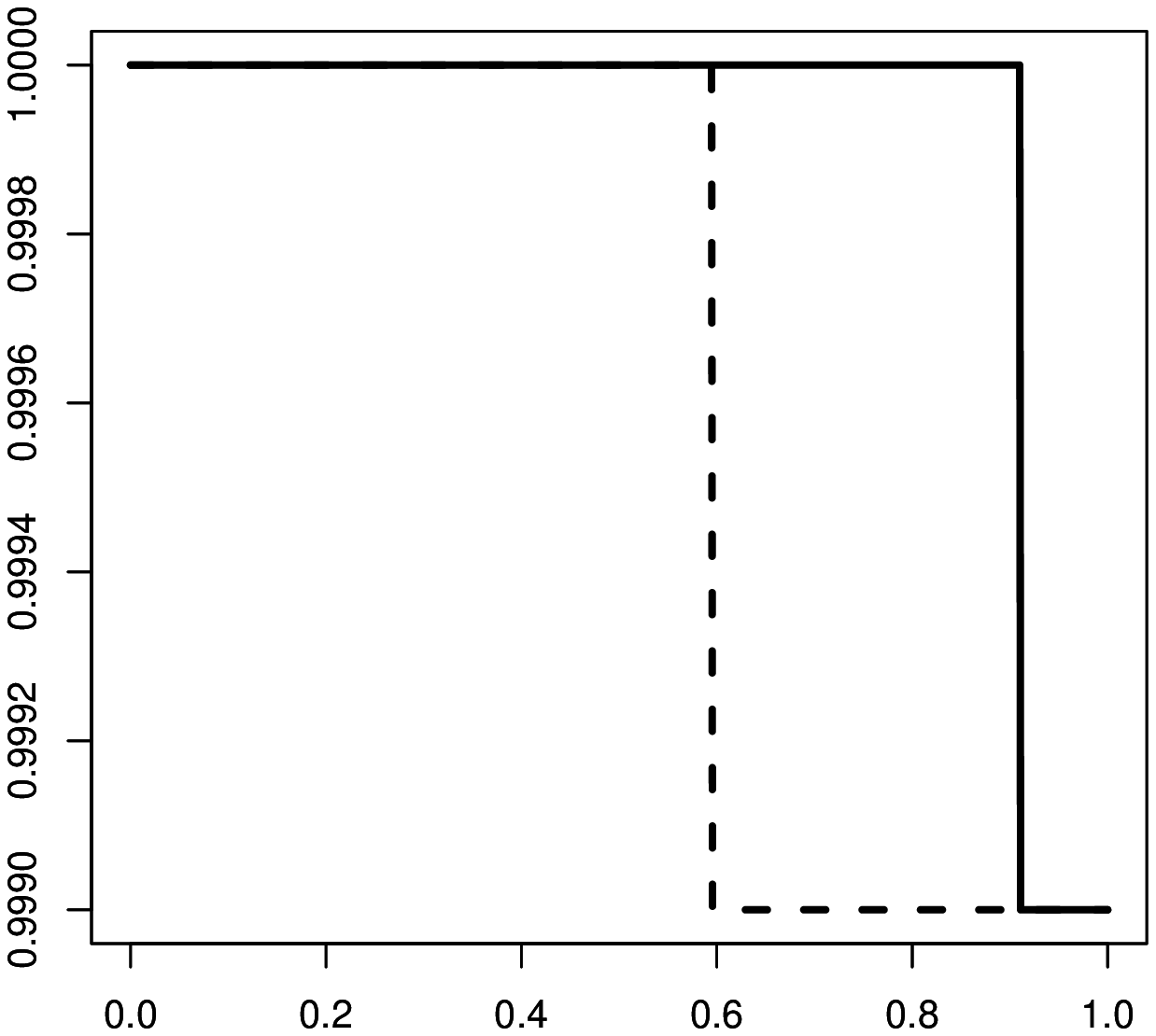}
}
\hspace{1cm}
\subfigure{
\includegraphics[width=0.4\textwidth]{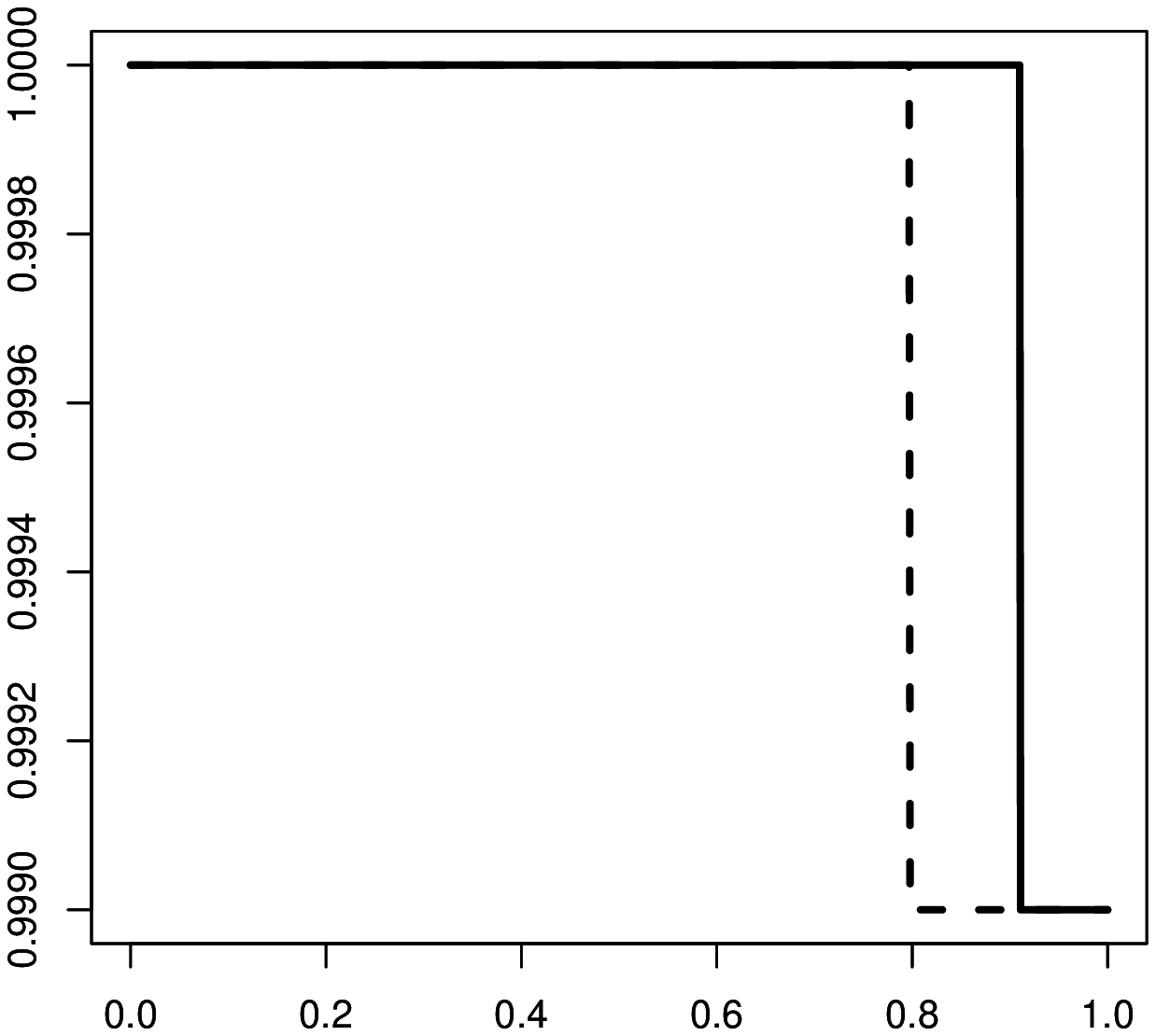}
}
%\vspace{-0.5cm}
\caption{Top: Optimal percentage of the initial budget invested in the risky asset as a function of $\beta$ for $\ell=10\%$ (left) and $\ell=20\%$ (right) and for $r=95\%$ (dotted), $r=97.5\%$ (dashed), $r=99\%$ (solid). Bottom: Recovery probability as a function of $\lambda$ for $\ell=10\%$ (left) and $\ell=20\%$ (right) and for $\beta=0.5\%$ and $r=99\%$ under $\avar$ (dashed) and $\reavar$ (solid).}
\label{fig: optimal portfolio}
\end{figure}
As a consequence, we obtain the following picture about optimal portfolios under $\reavar$ as a function of the parameters $\beta$ and $r$ that determine the level function $\gamma$. On the one hand, if $\beta\ge0.9\%$, then problem \eqref{eq: illustration 1 reavar} admits a unique optimal solution given by $x=1$. In this case, there is no difference between $\avar$ and $\reavar$ and the optimal portfolio in both cases corresponds to investing the whole budget in the risky asset. On the other hand, if $\beta<0.9\%$, then
problem \eqref{eq: illustration 1 reavar} admits the unique optimal solution
\[
x = \min\left\{\frac{(1-r)\ell}{\avar_\beta(R^2)-\avar_\alpha(R^2)},1\right\}.
\]
In particular, we observe that $x<1$ if and only if
\[
(1-r)\ell<\avar_\beta(R^2)-\avar_\alpha(R^2)=
\begin{cases}
0.0405 & \mbox{if} \ \beta\in(0,0.1\%],\\
0.000045\frac{1}{\beta}-
0.0045 & \mbox{if} \ \beta\in(0.1\%,0.9\%).
\end{cases}
\]
As a result, the optimal proportion of the budget invested in the risky asset depends on the relative size of the recovery parameters $\beta$ and $r$ as well as of the liability parameter $\ell$. Everything else remaining equal, the optimal portfolio weight for the risky asset is increasing in $\beta$ and $\ell$ while it is decreasing in $r$, as one would expect. In Figure \ref{fig: optimal portfolio} (top) we display the optimal percentage of the initial budget invested in the risky asset as a function of $\beta$ for different choices of $r$ and $\ell$. More precisely, we consider the situation where management targets a recovery of $95\%$, $97.5\%$, or $99\%$ of liabilities and where the size of liabilities amounts to $10\%$ or $20\%$ of the entire budget. In each of these situations there are realistic choices of $\beta$ under which, differently from the $\avar$ case, it is not optimal to fully invest in the risky asset. In fact, there are situations where a considerable size of the budget is optimally invested in the risk-free asset.

\smallskip
\begin{comment}
\begin{figure}[t]
\vspace{-0.8cm}
\centering
\subfigure{
\includegraphics[width=0.4\textwidth]{RecoveryProbability1}
}
\hspace{1cm}
\subfigure{
\includegraphics[width=0.4\textwidth]{RecoveryProbability2}
}
\vspace{-0.5cm}
\caption{Recovery probability as a function of $\lambda$ for $\ell=10\%$ (left) and $\ell=20\%$ (right) and for $\beta=0.5\%$ and $r=99\%$ under $\avar$ (blue) and $\reavar$ (black).}
\label{fig: recovery probability}
\end{figure}
\end{comment}
We complement the previous analysis by assessing the ability of $\avar$ and $\reavar$ to cover a pre-specified portion of liabilities when budget is invested optimally and capital is adjusted to respect regulatory requirements. To this effect, we denote the underlying probability by $\probp$ and compute for difference choices of $\lambda\in(0,1)$ the recovery probability
\[
\probp(\lambda) := \probp(b(1-x^\ast)(1+R^1)+bx^\ast(1+R^2)+\rho^\ast\ge\lambda L_1)
\]
where $x^\ast$ is the optimal percentage of the budget invested in the risky asset and
\[
\rho^\ast:=
\begin{cases}
\avar_\alpha(E_1(x^\ast)),\\
\reavar_\gamma(E_1(x^\ast),L_1).
\end{cases}
\]
We know that $x^\ast=1$ under $\avar$. In this case, one can easily show that for every $\lambda\in(0,1)$
\[
\probp(\lambda)=\probp(R^2+\avar_\alpha(R^2)+\ell\ge\lambda\ell)=
\begin{cases}
100\% & \mbox{if} \ \lambda\le1-\frac{0.0405}{\ell},\\
99.9\% & \mbox{if} \ \lambda>1-\frac{0.0405}{\ell}.
\end{cases}.
\]
If we work under $\reavar$, we obtain for every $\lambda\in(0,1)$
\[
\probp(\lambda)=\probp(x^\ast R^2+\ell+\max\{x^\ast\avar_\alpha(R^2),x^\ast\avar_\beta(R^2)-(1-r)\ell\}\ge\lambda\ell).
\]
In Figure \ref{fig: optimal portfolio} (bottom) we plot recovery probabilities under $\avar$ and $\reavar$ for recovery parameters $\beta=0.5\%$ and $r=99\%$. By definition of $R^2$, the recovery probability is at least $99.9\%$ in both cases. If $\ell=10\%$, then the optimal portfolio weight for the risky asset is $x^\ast\approx0.2$. In this case, $\avar$ guarantees full recovery up to $60\%$ of liabilities whereas $\reavar$ performs much better by ensuring full recovery up to $90\%$ of liabilities. The gap is narrower but still clear when $\ell=20\%$, in which case the optimal portfolio weight for the risky asset is $x^\ast\approx0.4$. In this case, $\avar$ guarantees full recovery up to $80\%$ of liabilities while $\reavar$ continues to ensure full recovery up to $90\%$ of liabilities.

%%%%%%%%%%%%%%%%%%%%%%%%

\subsection{Case study 2: Optimal portfolio with dependence uncertainty}
\label{sect: efficient frontier}

We consider a financial institution with total budget $b > 0$ at the initial date. The management
can invest in two assets with one-period relative returns $R^1$ and $R^2$ with the following characteristics:
\begin{itemize}
    \item $R^1$ has a normal distribution with mean $0$ and standard deviation $1.5\%$.
    \item $R^2$ has a Student distribution with mean $0.5\%$, scale $1\%$, and $2$ degrees of freedom.
    \item $R^1$ and $R^2$ have a Student copula with linear correlation $0.2$ and $2$ degrees of freedom.
\end{itemize}
The second asset is clearly riskier as it has infinite variance. At the terminal date the institution is exposed to deterministic liabilities amounting to $L_1 = b\ell$ for given $\ell\in(0,1)$. We denote by $x$ the fraction of total budget that is invested in the second asset. The corresponding end-of-period net asset value is therefore equal to
\[
E_1(x) := b(1-x)(1+R^1)+bx(1+R^2)-L_1 = b(1+R^1+x(R^2-R^1)-\ell).
\]
We study robust portfolio optimization under box uncertainty using the notation introduced in Section~\ref{sect: box uncertainty}. As a first step, we apply Monte Carlo simulation to generate a sample of $m=50000$ realizations of the random vector $(R^1,R^2)$, which are denoted by $(R^1_j,R^2_j)$ for $j=1,\dots,m$, and fix a benchmark probability measure $\probp\in\cP$ under which $(R^1,R^2)$ is discrete and satisfies
\[
\probp((R^1,R^2)=(R^1_j,R^2_j))=\frac{1}{m}=0.002\%.
\]
For a given $C\in[0,\frac{1}{m}]$ we set $\underline{\bm{\varepsilon}}:=(-C,\dots,-C)$ and $\overline{\bm{\varepsilon}}:=(C,\dots,C)$ and consider the corresponding perturbation set $\cE$. In particular, observe that for every $\bm{\varepsilon}\in\cE$ we have
\[
\probp_{\bm{\varepsilon}}((R^1,R^2)=(R^1_j,R^2_j))\in\left[\frac{1}{m}-C,\frac{1}{m}+C\right], \ \ \ j=1,\dots,m.
\]
The degree of box uncertainty therefore increases with the parameter $C$. In particular, the case $C=0$ corresponds to no box uncertainty. As in the previous case study, set $\alpha=1\%$ and for $\beta\in(0,\alpha)$ and $r\in(0,1)$ consider the level function given by
\[
\gamma(\lambda):=
\begin{cases}
\alpha & \mbox{if} \ \lambda\in[r,1],\\
\beta & \mbox{if} \ \lambda\in[0,r).
\end{cases}
\]
For given $\mu>0$ define the set of admissible holdings in the second asset by
\[
\cX(\mu) := \bigg\{x\in[0,1] \,; \ \min_{\bm{\varepsilon}\in\cE}\{(1-x)\E_{\probp_{\bm{\varepsilon}}}(R^1)+x\E_{\probp_{\bm{\varepsilon}}}(R^2)\}\ge\mu\bigg\}.
\]
We focus on the robust optimization problem
\begin{align}
\label{eq: case study 2}
\min_{x\in\cX(\mu)}\max_{\bm{\varepsilon}\in\cE}\reavar^{\probp_\varepsilon}_\gamma(E_1(x),L_1).
\end{align}
Filtering out the constant $b$ and using the explicit shape of $\gamma$ we equivalently obtain
\begin{align*}
\min_{x\in\cX(\mu)}\max_{\bm{\varepsilon}\in\cE}\max\{\avar^{\probp_{\bm{\varepsilon}}}_\alpha((1-x)R^1+xR^2-\ell),\avar^{\probp_{\bm{\varepsilon}}}_\beta((1-x)R^1+xR^2-r\ell)\}.
\end{align*}
To determine the efficient frontier we rely on the minimax identity established in Theorem~\ref{thm:minimax robust box} and adapt the approach of \ci{ZhuFu09}, which was applied to mean-risk problems under $\avar$. First, it follows immediately from Theorem~\ref{thm:minimax robust box}
that problem \eqref{eq: case study 2} is equivalent to
\begin{align*}
\min \ T & \\
s.t. & \ \ \
\frac{1}{\beta}\left(\frac{1}{m}\sum_{j=1}^mu^1_j+\max_{\bm{\varepsilon}\in\cE}\sum_{j=1}^m\varepsilon_ju^1_j\right)-v^1\le T,\\
& \ \ \ \frac{1}{\alpha}\left(\frac{1}{m}\sum_{j=1}^mu^2_j+\max_{\bm{\varepsilon}\in\cE}\sum_{j=1}^m\varepsilon_ju^2_j\right)-v^2\le T,\\
& \ \ \ u^1_j\ge0, \ u^1_j\ge v^1-(1-x)R^1_j-xR^2_j+r\ell, \ \ \ j=1,\dots,m,\\
& \ \ \ u^2_j\ge0, \ u^2_j\ge v^2-(1-x)R^1_j-xR^2_j+\ell, \ \ \ j=1,\dots,m,\\
& \ \ \ \frac{1}{m}\sum_{j=1}^m((1-x)R^1_j+xR^2_j)+\min_{\bm{\varepsilon}\in\cE}\left\{\sum_{j=1}^m((1-x)\varepsilon_jR^1_j+x\varepsilon_jR^2_j)\right\}\ge\mu,\\
& \ \ \ T,v^1,v^2\in\R, \ x\in[0,1].
\end{align*}
We now follow \ci{ZhuFu09}, to which we refer for the necessary details, and use the explicit form of $\cE$ (remember that $\underline{\bm{\varepsilon}}=(-C,\dots,-C)$ and $\overline{\bm{\varepsilon}}=(C,\dots,C)$ for some $C\in[0,\frac{1}{m}]$) to rewrite the inner optimizations over $\cE$ as convenient linear programs. With respect to the maximization over $\cE$, it suffices to observe that, by duality, for any given $u\in\R^m$ we have
\[
\max_{\bm{\varepsilon}\in\cE}\sum_{j=1}^m\varepsilon_ju_j = \min_{(z,\bm{\sigma},\bm{\tau})\in\R\times\R^m_+\times\R^m_-}\left\{C\sum_{j=1}^m(\sigma_j-\tau_j) \,; \ z+\sigma_j+\tau_j=u_j, \ j=1,\dots,m\right\}.
\]
Similarly, with respect to the minimization over $\cE$, for fixed $x\in[0,1]$ we obtain by duality that
\[
\min_{\bm{\varepsilon}\in\cE}\left\{\sum_{j=1}^m((1-x)\varepsilon_jR^1_j+x\varepsilon_jR^2_j)\right\} = \max_{(\bm{\zeta},\bm{\eta})\in\cS(x)}\left\{C\sum_{j=1}^m(\zeta_j-\eta_j)\right\},
\]
where the dual domain is defined by
\[
\cS(x) := \{(\bm{\zeta},\bm{\eta})\in\R^m_-\times\R^m_+ \,; \ \exists w\in\R, \ w+\zeta_j+\eta_j=(1-x)R^1_j+xR^2_j, \ j=1,\dots,m\}.
\]
In addition, we have $\cX(\mu)=\cX^\ast(\mu)$ where
\[
\cX^\ast(\mu) := \left\{x\in[0,1] \,; \ \exists (\bm{\zeta},\bm{\eta})\in\cS(x), \ \frac{1}{m}\sum_{j=1}^m((1-x)R^1_j+xR^2_j)+C\sum_{j=1}^m(\zeta_j-\eta_j)\ge\mu\right\}.
\]
\begin{figure}[t]
%\vspace{-0.8cm}
\centering
\subfigure{
\includegraphics[width=0.4\textwidth]{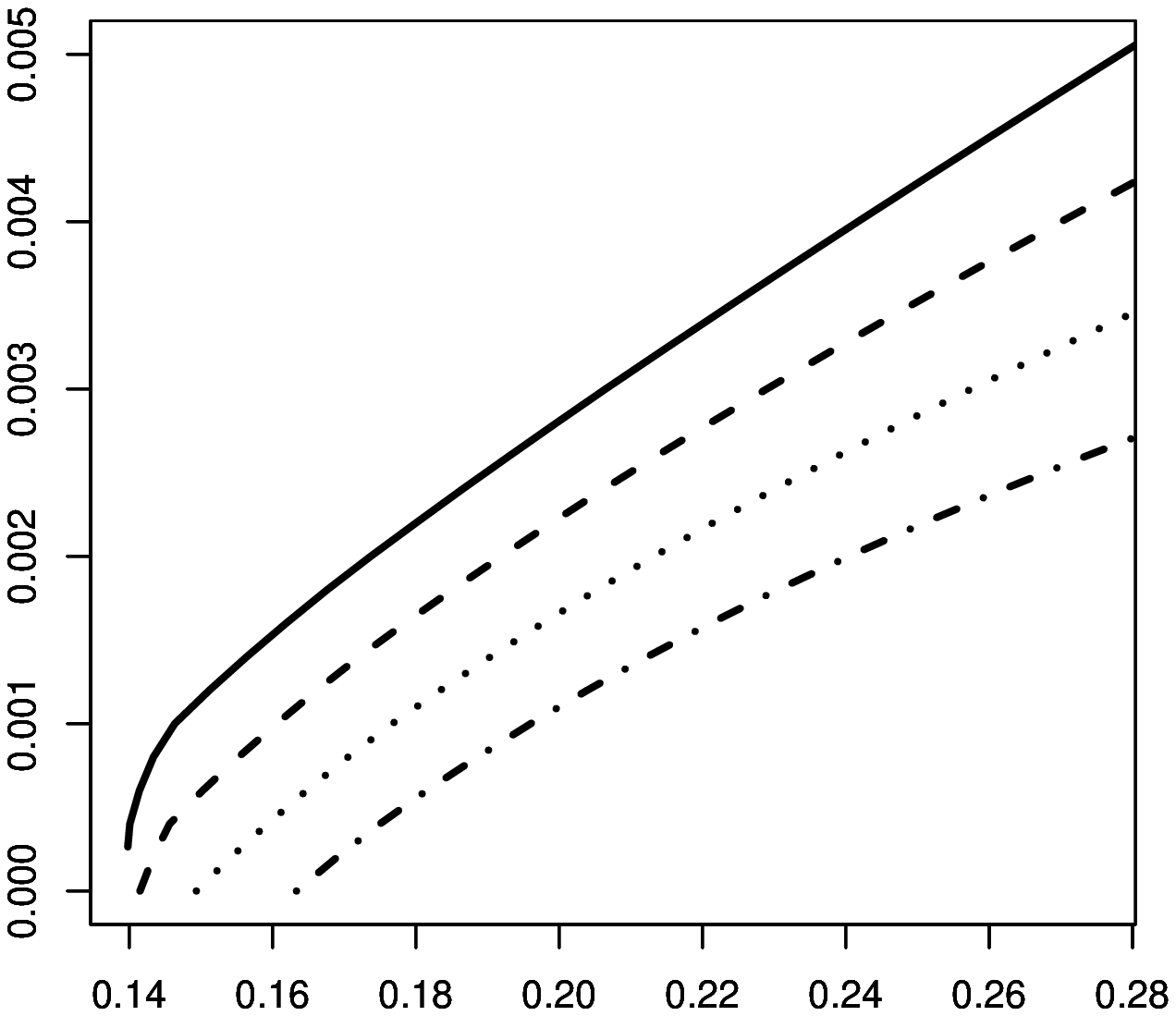}
}
\hspace{1cm}
\subfigure{
\includegraphics[width=0.4\textwidth]{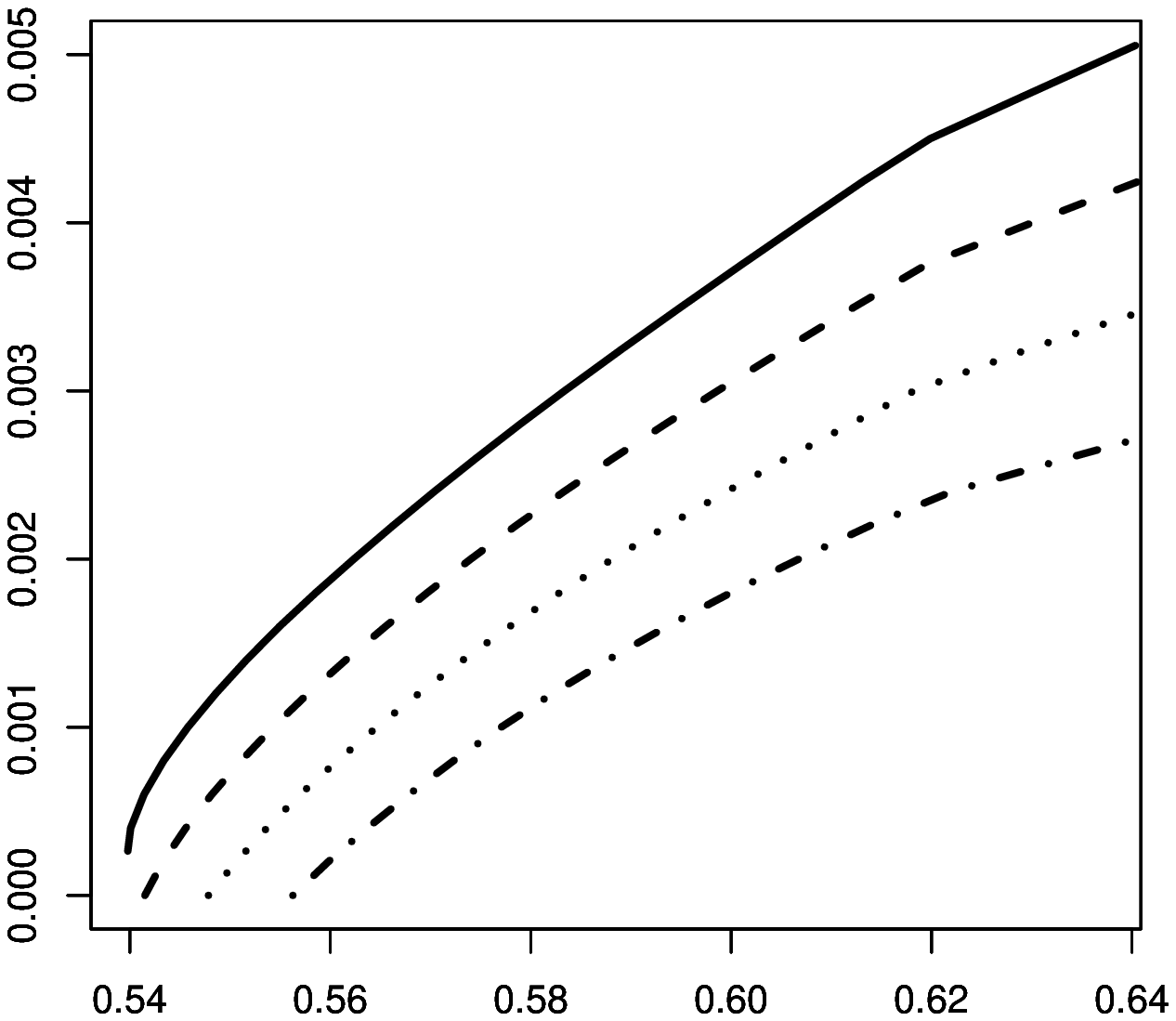}
}
\caption{Efficient frontier with $\reavar$ with $\beta=0.5\%$ and $r=90\%$ on the x-axis and expected relative returns on the y-axis for $\ell=10\%$ (left) and $\ell=50\%$ (right) and for $C=0$ (solid), $C=0.0001\%$ (dashed), $C=0.0002\%$ (dotted), $C=0.0003\%$ (dotdashed).}
\label{fig: efficient frontier}
\end{figure}
As a consequence, problem \eqref{eq: case study 2} is equivalent to the tractable linear program
\begin{align*}
\min \ T & \\
s.t. & \ \ \ \frac{1}{\beta m}\sum_{j=1}^mu^1_j+\frac{C}{\beta}\sum_{j=1}^m(\sigma^1_j-\tau^1_j)-v^1\le T,\\
& \ \ \ \frac{1}{\alpha m}\sum_{j=1}^mu^2_j+\frac{C}{\alpha}\sum_{j=1}^m(\sigma^2_j-\tau^2_j)-v^2\le T,\\
& \ \ \ u^1_j\ge v^1-(1-x)R^1_j-xR^2_j+r\ell, \ \ \ j=1,\dots,m,\\
& \ \ \ u^2_j\ge v^2-(1-x)R^1_j-xR^2_j+\ell, \ \ \ j=1,\dots,m,\\
& \ \ \ u^i_j\ge0, \ u^i_j=z^i+\sigma_j+\tau_j, \ \ \ i=1,2, \ j=1,\dots,m,\\
& \ \ \ z^i\in\R, \ \sigma^i_j\ge0, \ \tau^i_j\le0, \ \ \ i=1,2, \ j=1,\dots,m,\\
& \ \ \ \frac{1}{m}\sum_{j=1}^m((1-x)R^1_j+xR^2_j)+C\sum_{j=1}^m(\zeta_j-\eta_j)\ge\mu,\\
& \ \ \ w+\zeta_j+\eta_j=(1-x)R^1_j+xR^2_j, \ \ \ j=1,\dots,m,\\
& \ \ \ w\in\R, \ \zeta_j\le0, \ \eta_j\ge0, \ \ \ j=1,\dots,m,\\
& \ \ \ T,v^1,v^2\in\R, \ x\in[0,1].
\end{align*}
A standard dual simplex method can be employed to find the optimal value $\rho^\ast(\mu)$ of this linear program as a function of a target expected relative return $\mu$. In Figure~\ref{fig: efficient frontier} we plot the corresponding efficient frontier, i.e., the set of points $(\rho^\ast(\mu),\mu)$, for a given range of expected relative returns. We focus on the two situations where liabilities amount to $10\%$, respectively $50\%$, of the initial budget. As is intuitive, in the latter case the same level of expected relative returns is achieved at the cost of higher risk. A direct inspection of the plots reveal that, for the chosen range of target returns, risk is $2$ to $4$ times higher when liabilities have a larger size. However, the qualitative impact of dependence uncertainty, in the form of box uncertainty, is the same in both situations. If the degree of uncertainty increases, the same level of expected relative returns is achieved at the cost of higher risk. Interestingly enough, the impact of dependence uncertainty on risk is more pronounced when the size of liabilities is smaller.

%%%%%%%%%%%%%%%%%%%%%%%%%%%%%%%%%%%%%

\appendix

\section{A minimax theorem}

In this appendix we record a special minimax theorem that we apply repeatedly in the paper. For $n\in\N$ we denote by $\Delta ^n$ the $n$-dimensional simplex, i.e., the set of vectors in $\R^n$ with nonnegative components summing up to $1$.

\begin{theorem}
\label{theo: minimax}
Let $f_1,\dots,f_n:\R\to\R$ be convex functions attaining their minimum on (nonempty) compact intervals $\cI_1,\dots,\cI_n\subset\R$. Then,
\begin{equation}
\label{eq: minimax}
\max_{i=1,\dots,n}\min_{x\in\R}f_i(x) = \min_{\bm{x}\in\R^n}\max_{i=1,\dots,n}f_i(x_i).
\end{equation}
\end{theorem}
\begin{proof}
Set $\cI=\cI_1\times\cdots\times\cI_n$. For every $\boldsymbol{x}\in \R^n$ we can write
\begin{align*}
\max_{i=1,\dots,n} \min_{x\in \R} f_i(x) = \max_{\boldsymbol{\theta} \in \Delta ^n} \sum_{i=1}^{n} \theta_i  \min_{x\in \R} f_i(x) = \max_{\boldsymbol{\theta} \in \Delta ^n}\min_{\boldsymbol{x}\in \R^n} \sum_{i=1}^{n} \theta_i f_i(x_i) = \max_{\boldsymbol{\theta} \in \Delta ^{n}}\min_{\boldsymbol{x}\in\cI} \sum_{i=1}^{n} \theta_i f_i(x_i).
\end{align*}
As $\Delta ^n$ and $\cI$ are compact and convex and the function $(\boldsymbol{\theta},\boldsymbol{x})\mapsto\sum_{i=1}^{n} \theta _i f_i(x_i)$ is linear in $\boldsymbol{\theta}$ and convex in $\boldsymbol{x}$, we can apply the classical minimax theorem in \ci{fan1953minimax} to infer that
\begin{align*}
\max_{i=1,\dots,n} \min_{x\in \R} f_i(x) = \min_{\boldsymbol{x}\in\cI}\max_{\boldsymbol{\theta} \in \Delta ^{n}}\sum_{i=1}^{n} \theta_i f_i(x_i) = \min_{\boldsymbol{x}\in\cI}\max_{i=1,\dots,n}f_i(x_i) = \min_{\boldsymbol{x}\in\R^n}\max_{i=1,\dots,n}f_i(x_i).
\end{align*}
This delivers the desired minimax equality.
\end{proof}

\smallskip

\begin{remark}
The outer minimum in the minimax equality \eqref{eq: minimax} cannot be taken over $\R$ in general. To see this, consider the convex functions $f_1,f_2:\R\to\R$ defined by
\[
f_1(x):=
\begin{cases}
-x-2 & \mbox{if} \ x\le-2,\\
0 & \mbox{if} \ -2<x\le-1,\\
x+1 & \mbox{if} \ x>-1,
\end{cases}
 \ \ \ \
f_2(x):=
\begin{cases}
-x+1 & \mbox{if} \ x\le1,\\
0 & \mbox{if} \ 1<x\le2,\\
x-2 & \mbox{if} \ x>2.
\end{cases}
\]
Note that $f_1$ and $f_2$ attain their minimum on $[-2,-1]$ and $[1,2]$, respectively. However,
\[
\max_{i=1,2}\min_{x\in\R}f_i(x) = 0 < 1 = \min_{x\in\R}\max_{i=1,2}f_i(x).
\]
\end{remark}

\bibliography{bibtex}
\bibliographystyle{jmr}

\begin{comment}

\end{comment}

\end{document}